\documentclass[copyright,creativecommons]{eptcs}
\providecommand{\event}{EXPRESS/SOS 2024} 
\usepackage{breakurl}

\usepackage{amsmath}
\usepackage{amssymb}
\usepackage{stmaryrd}
\usepackage{amsthm}
\usepackage{float}

\usepackage[capitalise]{cleveref}
\usepackage{xspace}
\usepackage{ebproof}

\usepackage[acronym,nohypertypes={acronym}]{glossaries}

\usepackage{mathtools}

\usepackage{xfrac}
\usepackage{multicol}

\newcommand{\butf}{\textsc{ButF}\xspace}
\newcommand{\futhark}{\textsc{Futhark}\xspace}
\newcommand{\pic}{$\pi$-calculus\xspace}
\newcommand{\apic}{applied \pic}
\newcommand{\epic}{E$\pi$\xspace}
\newcommand{\mra}{\rightarrow}
\newcommand{\pizero}{\mathbf{0}}
\newcommand{\pisend}[2]{\overline{#1}\langle#2\rangle}
\newcommand{\pibroad}[2]{\overline{#1}{:}\langle#2\rangle}
\newcommand{\pirecv}[2]{#1(#2)}
\newcommand{\translate}[1]{\llbracket{#1}\rrbracket}
\newcommand\pitrans[1]{\translate{#1}}
\newcommand{\transtopi}[1]{\translate{#1}}

\newcommand{\ba}{\mathbin{\dot\approx_a}}

\renewcommand\mit[1]{\mathit{#1}}

\newcommand{\ite}[4]{\ensuremath{[#1 \bowtie #2] \, #3, \, #4}}

\newtheorem{theorem}{Theorem}
\newtheorem{definition}{Definition}
\newtheorem{lemma}{Lemma}

\newcommand\piifthenelse[3]{\left[#1\right]\:#2,#3}
\newcommand\piifthen[2]{\left[#1\right]\:#2}

\newcommand\mraw{\xrightarrow{\circ}}
\newcommand\pisubst[2]{\left\{\sfrac{#2}{#1}\right\}}
\newcommand{\opcor}{\gtrless_{ok}}

\usepackage{breakurl}
\newacronym{GPGPU}{GPGPU}{general purpose \gls{GPU}}
\newacronym{GPU}{GPU}{graphics processing unit}
\newacronym{SOAC}{SOAC}{second-order array combinator}
\newacronym{LMAD}{LMAD}{linear-memory access descriptor}
\newacronym{BFS}{BFS}{breadth-first search}
\newacronym{CPU}{CPU}{central processing unit}
\newacronym{IR}{IR}{intermediate representation}
\newacronym{SIMT}{SIMT}{single instruction multiple thread}
\newacronym{SIMD}{SIMD}{single instruction multiple data}

\begin{document}
\title{Functional Array Programming in an Extended Pi-Calculus}

\author{Hans Hüttel \institute{Department of Computer Science,
    University of Copenhagen, Denmark (\texttt{hans.huttel@di.ku.dk})} \and Lars Jensen\institute{Department of Computer
    Science, Aalborg University, Denmark (\texttt{larsdjand@gmail.com})} \and Chris Oliver
  Paulsen\institute{Department of Computer Science, Aalborg
    University, Denmark (\texttt{chris@coppm.xyz})} \and
  Julian Teule \institute{Department of Computer Science, Aalborg University, Denmark (\texttt{julian@jtle.dk}})}

\newcommand{\size}{\textsf{size}\xspace}
\newcommand{\reduce}{\textsf{reduce}\xspace}
\newcommand{\iot}{\textsf{iota}\xspace}
\newcommand{\map}{\textsf{map}\xspace}
\newcommand{\runa}[1]{\textsc{(#1)\xspace}}
\newcommand{\tup}{\texttt{tup}}

\def\authorrunning{H. Hüttel, L. Jensen, C.O. Paulsen \& J. Teule}
\def\titlerunning{Functional Array Programming in an Extended Pi-Calculus}
\maketitle

\begin{abstract}
  We study the data-parallel language \butf, inspired by the \futhark
  language for array programming. We give a translation of \butf into
  a version of the \pic with broadcasting and labeled names. The
  translation is both complete and sound. Moreover, we propose a
  cost model by annotating translated \butf processes.  This is
  used for a complexity analysis of the translation.
\end{abstract}

\section{Introduction}

The \futhark programming language is a functional language whose goal
is to abstract parallel array operations by means of utilizing
\emph{second order array combinators}, such as \map and
\reduce \cite{FutharkPhd}. The \futhark compiler then
efficiently translates code into optimized code for the targeted
hardware.

Parallel hardware, such as \glspl{GPU}, does not support arbitrary
nesting of parallel operations, nor arbitrarily large problem sizes,
and the \futhark compiler therefore produces a program for which the
outermost levels of nested operations of a program are executed in
parallel.

The \gls{GPU} programs produced by the \futhark compiler are therefore
limited by the physical constraints of the hardware in question, and
it would thus be interesting to analyze \futhark programs in the
setting of an underlying parallel language without these limitations.

It is known that there exist sound translations of the
$\lambda$-calculus and different reduction strategies into the simple
$\pi$-calculus
\cite{DBLP:books/daglib/0098267,DBLP:books/daglib/0004377}. Milner was
the first to provide such a translation
\cite{pi_functions_as_processes} and Sangiorgi extended his work
\cite{DBLP:conf/mfps/Sangiorgi93,DBLP:journals/iandc/Sangiorgi94,DBLP:journals/mscs/Sangiorgi99}. These
encodings identify the essence of how to implement a functional
programming language on a parallel architecture using references in
the form of name-passing and the ability to express arbitrary levels
of nested concurrency and parallelism.

In this paper we use this work as the inspiration for a translation of
a functional array programming language which is a subset of \futhark
into an extended $\pi$-calculus, $E\pi$. In $E\pi$ we extend the
setting to one containing structured data
\cite{pi_mobile_values_secure_communication,DBLP:conf/lics/BengtsonJPV09}
and broadcasting, as these are central to the protocol used by
\futhark.

Our focus is on how to encode the array structure and a subset of
second-order array operators from \futhark into $E\pi$. For the proof
of operational correspondence we use a coinductive approach which
lends itself well to expressing the correctness of our encoding.  Our
approach is inspired by that of Amadio et
al. \cite{DBLP:conf/fct/AmadioLT95} in that we distinguish between the 
``important'' and ``administrative'' computation steps. This also
allows us to compare the cost of the translation to that of \futhark
constructs.

% The present paper is a continuation of the MSc dissertation by the
% first three authors \cite{p10} and thus uses it as a basis for proofs
% of the results given here.

\section{A language for array programming} %Formal definition

First we introduce \butf and the process calculus \epic that
will be the target language of our translation.

\subsection{Basic Untyped \futhark}

Basic Untyped \futhark (\butf) deals only with functional array
computation and omits the module system of \futhark. \butf is thus a simple
$\lambda-$calculus with arrays, tuples, and binary
functions.

\subsubsection{Expressions in \butf}

The formation rules of \butf expressions are shown below.
\begin{equation}
  \begin{aligned}
    e &::= \ b\mid x\mid [e_1,\dots,e_n]\mid e_1[e_2]\mid\lambda
  x.e_1\mid e_1\ e_2\mid  (e_1,\dots,e_n) \mid \texttt{if}\ e_1\ \texttt{then}\ e_2\ \texttt{else}\ e_3\\
    b&::= n\mid \texttt{map}\mid\texttt{iota}\mid\texttt{size} \mid\odot
\end{aligned}
\label{eq:syntax}
\end{equation}
\butf makes use of prefix application $e_1\ e_2$. Constants are
denoted as $b$, and are integer constants, arithmetic operations
$\odot$ and the array operations described in
\cref{sec:arrayop}. Arrays are denoted by $[e_1,\dots,e_n]$ and tuples
are denoted by $(e_1,\dots,e_n)$. The expression $e_1[e_2]$ will
evaluate to the place in the array $e_1$ whose index is the value of
$e_2$. To express a unary tuple, we use the notation $(x,)$, while a
empty tuple is denoted as $()$.

\butf is a call-by-value language whose values $v \in V$ are
constants, function symbols and arrays and tuples that contain values
only.
\[ v ::= b \mid [v_1,\ldots v_n] \mid (v_1,\ldots,v_n) \]
The semantics of \butf is given by the reduction relation $\mra$, and
reductions are of the form $e \mra
e'$. \Cref{eq:application_semantics} shows the semantics of
application is beta-reduction. % the three forms of application.
\begin{align}\label{eq:application_semantics}
    \frac{}{(\lambda x.e)\ v\mra e\{x\mapsto v\}}
 % & &
%    \frac{(\lambda x.e)\ v_1\mra e'}{(\lambda x.e){:}(v_1,v_2,\dots)\mra e'{:}(v_2,\dots)}
% & &
%  \frac{}{e.()\mra e}
\end{align}
Arrays contain elements that can be arbitrary expressions.
\Cref{eq:array_semantics} shows how each subexpression in an array can
take a reduction step. Fully evaluated expressions can be indexes with
the index operator.
\begin{align}\label{eq:array_semantics}
      \frac{e_i\mra e_i' \quad 1 \leq i \leq
      n}{[e_1,\dots,e_i,\dots,e_n]\mra[e_1,\dots,e_i',\dots,e_n]}
    & &
      \frac{0\leq i\leq n-1}{[v_1,\dots,v_n][i]\mra v_{i+1}}
\end{align}

Lastly, we have the conditional structure that allows branching depending on the result of $e_1$.
\begin{equation}
    \begin{gathered}
        \begin{prooftree}
            \hypo{v\neq 0}
            \infer1{\texttt{if}\ \mit{v}\ \texttt{then}\ e_2\ \texttt{else}\ e_3\rightarrow e_2}
        \end{prooftree}\qquad
        \begin{prooftree}
            \hypo{v = 0}
            \infer1{\texttt{if}\ \mit{v}\ \texttt{then}\ e_2\ \texttt{else}\ e_3\rightarrow e_3}
        \end{prooftree}
    \end{gathered}
\end{equation}

\subsubsection{Array Operations}\label{sec:arrayop}

\butf uses the array operations \size, \iot and \map.  These have been
chosen since they can be used to define other common array operators
such as \textsf{concat}, \textsf{reduce}, and \textsf{scan}
\cite{p10}. This allows us to simplify the translation and the proof
of its correctness.

The intended behaviour of the function constants is as follows. \size
receives a handle of an array and returns its element count and \iot
creates an array of the size of its parameter with values equal to the
values' index. The \map function allows for applying a function to
each element in an array. 

%Futhark has extensive collection of array operations which make it a usefull array programming language.
%For \butf we have reduced the operations to \map, \texttt{size}, and \iot, as we believe that these can be used to construct the most common array operators of Futhark without a loss of assympthotic complexity.
%\map is the main interest in this work, as it allows for \gls{SIMD} operations on arrays.
%\iot is also interesting because it allows for the construction of arrays of arbirary size.

The reduction rules for the function constants are shown below. Notice
that \map is uncurried -- it cannot be partially applied. This
eliminates the translation case of a partially applied map function.
\begin{align}\label{eq:soac_semantics}
      \map\ ((\lambda x.e),[v_1,\dots,v_n])\mra[e\{x\mapsto v_1\},\dots,e\{x\mapsto v_n\}]
\\
     \size\ [v_1,\dots,v_n]\mra n
 \qquad
    \iot\ n\mra[0,1,2,\dots,n-1]
\end{align}

\subsection{Extended Pi-Calculus}
The language used as the target for the translation is the Extended
$\pi$-calculus (\epic), presented in previous work \cite{p10}, and is
based on the \apic presented by Abadi, Blanchet, and Fournet
\cite{DBLP:journals/jacm/AbadiBF18}.  This calculus is extended with
broadcast communication as presented by Hüttel and Pratas
\cite{DBLP:journals/corr/HuttelP16} as well as simple first order
composite names based on \cite{composite_names}.
\vspace{-2mm}
\subsubsection{Processes in \epic}

Processes are given by the formation rules below.
\begin{equation} \label{fig:Process_pi_syntax}
  \begin{aligned}
    P&::= \pizero \mid P|Q \mid {!}P\mid\nu a.P \mid A.P\mid\bullet P
    \mid \ite{M}{N}{P}{Q} & \qquad &
    A::= \pisend c{\vec T}\mid \pirecv c{\Vec x}\mid\pibroad c{\vec T}\\
    c&::= a\mid x\mid a\cdot I\mid x\cdot I &&
    I::= n\mid x\mid\texttt{all}\mid\texttt{tup}\mid\texttt{len}\\
    T&::= n \mid a\mid x\mid T\odot T
  \end{aligned}
\end{equation}
$T$ ranges over terms that can be sent on channels. These may be a
number ($n$), a channel name ($a$), or a variable ($x$).  A term may
also be a binary operation on two terms ($T \odot T$). These
operations are as in \butf, except that one cannot use them on
names. We let $u$ range over the set of variables and names.

Processes $P$ can be the empty process $\pizero$ which cannot reduce
further, parallel composition ($P \mid Q$) consisting of two processes
in parallel, replication (${!}{P}$) which constructs an unbounded
number of process $P$ in parallel, and declaration of new names
($\nu u.P$), which restricts $u$ to the scope of $P$.  A process
\ite{M}{N}{P}{Q} is a conditional process where
$\bowtie \, \in \{<, >, =,\neq\}$. If $M \bowtie N$, it proceeds as $P$ and
else as $Q$. Actions $A$ are output $\pisend c{\vec T}$ and input
$\pirecv c{\Vec{x}}$; in $\pirecv c{\Vec{x}}.P$, the variables
$\Vec{x}$ are bound in $P$. In $\nu a.P$, $a$ is bound in $P$. We let
$\mit{fn}(P)$ and $\mit{fv}(P)$ denote the sets of free names and free
variables in $P$. The process $\bullet P$ denotes that $P$ begins with
an important computation step; this is explained in Section
\ref{subsec:semantics} and is used in analyzing the complexity of our encoding.

Broadcasting in \epic is denoted as $\pibroad c{\vec T}$. It can send
a vector of terms $\vec{T}$ over a channel $c$ to multiple processes
in a single reduction, atomically. A channel name $c$ can be a name
($a$) is a composite name consisting of a name followed by an
identifier $I$ that can be either a number or a label. These labels
are used to distinguish between several different translation
constructs. In particular, in the encoding, labels describe if a
reduction involves an entire array (\texttt{all}), the reduction of a
tuple (\texttt{tup}) or the computation of the length of an array
(\texttt{len}).

\vspace{-2mm}

\subsubsection{Semantics} \label{subsec:semantics}

The structure of the semantics for \epic is similar to that of the
\pic, using a structural congruence relation that identifies process
expression with the same structure and a reduction relation.

\begin{figure}[H]
  \begin{center}
\begin{tabular}{llll}
    \runa{Rename}& $P \equiv P'$ by $\alpha$-conversion &
   \runa{Replicate}   & ${!}P \equiv P\mid{!} P$\\[3mm]
  \runa{Par-$\pizero$}  & $P\mid\pizero\equiv P$ &
                                                    \runa{New-$\pizero$}
                             & $\nu n.\pizero \equiv\pizero$\\[3mm]
    \runa{Par-A}  & $P\mid(Q\mid R) \equiv(P\mid Q)\mid R$ &
                                                             \runa{New-A}
                             & $\nu u.\nu v.P \equiv\nu v.\nu u.P$ \\[3mm]
    \runa{Par-B}  & $P\mid Q \equiv Q\mid P$ & \runa{New-B} &
                                                              $P\mid\nu
                                                              u.Q
                                                              \equiv\nu
                                                              u.(P\mid
                                                              Q)$ \\
                                       &&& when $u\not\in\mit{fv}(P)\cup\mit{fn}(P)$\\
\end{tabular}
\end{center}
    \caption{The structural congruence rules for the extended $\pi$ calculus }
\label{fig:Pi_struct_rules}
\end{figure}

The congruence rules shown in \cref{fig:Pi_struct_rules} are common for most $\pi$-calculi, and for a more detailed explanation see previous work \cite{pi_functions_as_processes}. 

\begin{figure}
\begin{center}
\begin{tabular}{llll}
\runa{Comm}  & $\pisend {C}v.P\mid\pirecv
               {C}x.Q\xrightarrow\tau P\mid Q\pisubst xv $ & \runa{Broad}
  & $\pibroad{C}{v}.Q \mid \pirecv{C}{x_1}.P_1\mid\dots\mid\pirecv
               c{x_n}.P_n $\\ &&& $\xrightarrow{:c}Q\mid P_1\pisubst {x_1}v\mid\dots\mid P_n\pisubst {x_n}v$ \\[3mm]
\runa{Par} & $\begin{prooftree}
            \hypo{P\xrightarrow\tau P'}
            \infer1{P\mid Q\xrightarrow\tau P'\mid Q}
        \end{prooftree}$ & \runa{B-Par} & $\begin{prooftree}
            \hypo{P\xrightarrow{:c} P'}
            \hypo{Q\not\downarrow_c}
            \infer2{P\mid Q\xrightarrow{:c} P'\mid Q}
          \end{prooftree}$  \\[4mm]
\runa{Res-1} & $\begin{prooftree}
            \hypo{P\xrightarrow{:c}P'}
            \hypo{c\not\in\{u, u\cdot I\}}
            \infer2{\nu u.P\xrightarrow{:c}\nu u.P'}
        \end{prooftree}$ & \runa{Res-2} & $\begin{prooftree}
            \hypo{P\xrightarrow{:c} P'}
            \hypo{c\in\{u, u\cdot I\}}
            \infer2{\nu u.P\xrightarrow{\tau}\nu u.P'}
          \end{prooftree}$ \\[5mm]
\runa{Struct} & $\begin{prooftree}
            \hypo{P\xrightarrow q P'}
            \infer1{Q\xrightarrow q Q'}
        \end{prooftree} \quad\text{if $P\equiv Q$ and $P'\equiv Q'$}$
                                                       & \runa{Then} &
                                                                       $\ite{M}{N}{P}{Q}\xrightarrow\tau
                                                                       P\quad\text{if
                                                                       $M\bowtie
                                                                       N$}$
  \\[4mm]
  \runa{Else} & $\ite{M}{N}{P}{Q}\xrightarrow\tau Q\quad\text{if $M\not\bowtie N$}$
\end{tabular}
\end{center}
    \caption{The reduction rules of extended processes in \epic. Here, $q$ is either $\tau$ or some ${:}b$.}
    \label{fig:pi_reduction_rules}
\end{figure}

The transition labels $\tau$ and ${:c}$ in \cref{fig:pi_reduction_rules} ensure that all parallel receivers of a broadcast are used in the broadcast.
A reduction arrow without a label, $\mra$, is used to denote an arbitrary reduction.

\begin{figure}
    \begin{equation*}
    \begin{gathered}
        \runa{Adm}\quad\begin{prooftree}
            \hypo{P\xrightarrow\tau P'}
            \infer1{P\mraw P'}
        \end{prooftree}\qquad
        \runa{NonAdm}\quad\begin{prooftree}
            \hypo{P \mra P'}
            \infer1{\bullet P \xrightarrow\bullet P'}
        \end{prooftree} \qquad
        \runa{Both}\quad\begin{prooftree}
            \hypo{P\xrightarrow s P'}
            \hypo{s\in\{\bullet,\circ\}}
            \infer2{P\mra P'}
        \end{prooftree}
        \end{gathered}
    \end{equation*}
    \caption{Labeled semantics for important ($\xrightarrow\bullet$) and administrative reductions ($\xrightarrow\circ$) in \epic.}
\label{fig:pi_reduction_rules_labeled}
\end{figure}
The semantics shown in \cref{fig:pi_reduction_rules_labeled} are used to
distinguish between important and administrative reductions. This will
be used in the translation to distinguish transitions which emulate a
\butf reduction, and transitions which facilitate the translation.

\subsubsection{Weak Bisimilarity}

Our notion of semantic equivalence is called weak administrative
barbed bisimilarity as is a form of barbed congruence
\cite{pi_a_tutorial}.  To define it, we use an observability predicate
$\downarrow_\alpha$ where $\alpha$ is $a$ or $\overline{a}$. If
$\pisend\alpha b.P \rightarrow P$ then $P \downarrow_\alpha$. The
definition (which involves broadcast) follows the structure of that of
\cite{DBLP:books/daglib/0098267}. The arrows $\xRightarrow\circ$ and
$\xRightarrow\bullet$ denote multiple transitions as follows.
\begin{definition}
    We define $\xRightarrow{s}$ as follows for the label $s\in\{\bullet,\circ\}$.
    \begin{equation*}
        \xRightarrow{s}\ =\left\{\begin{matrix}
            s=\circ&\xrightarrow{\circ}^\ast\\
            s=\bullet&\xrightarrow{\circ}^\ast\xrightarrow{\bullet}\\
        \end{matrix}\right.
    \end{equation*}
\end{definition}

In weak administrative barbed bisimilarity, important reductions in the
one process must be matched by important reductions in the other
process.

\begin{definition}[Weak Administrative Barbed
  Bisimulation]\label{def:wabb}
  A symmetric relation $R$ over processes is called a \emph{weak
    administrative barbed bisimulation} (wabb) if whenever
  $(P,Q) \in R$, the following holds
\begin{enumerate}
    \item If $P\xrightarrow{\bullet} P'$ then there exists a $Q'$ such that $Q \xRightarrow{\bullet} Q'$ and $(P',Q') \in R$,
    \item If $P\mraw P'$ then $Q\xRightarrow{\circ} Q'$ and $(P',Q')\in R$,
    \item For all contexts $C$, $(C[P],C[Q]) \in R$,
    \item For all prefixes $\alpha$, if $P\downarrow_\alpha$ then $Q \xRightarrow{\circ}\downarrow_\alpha$.
    \end{enumerate}
    
    We write $P\ba Q$ if there exists
    a weak administrative barbed bisimulation $R$ such that
    $(P,Q)\in R$.
\label{def:weak_barbed_bisimulation}
\end{definition}

\section{Translating \butf to \epic}

The translation from \butf into the extended \pic is very similar to
the approach of Robin Milner \cite{pi_functions_as_processes}.  We use
the same notation of $\translate{e}_o$ for the translation of the
\butf expression $e$ into a process emitting the representation of the
its value on the channel $o$.  Our translation differs in that \butf
uses not numbers but also arrays and the accompanying
operators as values.

\subsection{Translating the functional fragment}
First, we define the translation of the part of \butf\ that
corresponds to an applied $\lambda$-calculus -- numbers, functions,
and application, shown in \cref{eq:translate_normal_stuff}. 
Numbers and variables are themselves already evaluated, and they are thus sent directly on the out channel.
With abstractions, we introduce a function channel $f$, which represents that abstraction.
A replicated process is listening on $f$, waiting for other processes to call it.
An application consists of two subexpressions that must be evaluated before the function channel and value can be extracted on the two inner $o$ channels.

The translation has been annotated with $\bullet$ to ensure that
transitions in \butf are matched by a single bullet. This can be seen in
application, $\pitrans{e_1\ e_2}_o$, which requires a single
$\xrightarrow\bullet$ before the function is called.

\begin{figure}
  \begin{align*}
    \translate x_o &= \pisend ox\\
    \translate n_o &= \pisend on\\
    \translate{\texttt{if}\ e_1\ \texttt{then}\ e_2\ \texttt{else}\ e_3}_{o} &=\\
    \nu o_1.(\translate{e_1}_{o_1}&\mid\pirecv{o_1}{v}.\bullet\piifthenelse{v \neq 0}{\translate{e_2}_{o}}{\translate{e_3}_{o}})\\
    \translate{\lambda x.e}_o &=\nu f.(\pisend of\mid{!}\pirecv f{x,r}.\translate{e}_r)\\
    \translate{e_1\ e_2}_o &=\\
    \nu o_1.\nu o_2.(\translate{e_1}_{o_1}&\mid\translate{e_2}_{o_2}
    \mid\pirecv{o_1}f.\pirecv{o_2}v.\bullet\pisend f{v, o})\\
  \end{align*}
  \vspace{-14mm}
  \caption{The translation for basic expressions.}\label{eq:translate_normal_stuff}
\end{figure}

\subsection{Tuples}
Tuples are translated by evaluating all subexpressions in parallel
and waiting for them all to return on their out channels.  These
results are then all repeatedly sent on the $h$ channel.  Users of the
tuple can read the handle channel to get access to all the values.

Therefore tuple elements are sent
on $h\cdot \tup$ to ensure that the tuple can not be used in places that
expect arrays. By composing with the label $\tup$, the array can only be
accessed with this label and not the array labels $\texttt{all}$ and
$\texttt{len}$.
  \begin{align*}
    \translate{(e_1,\dots,e_n)}_o =\nu{o_1}.\dots.\nu{o_n}.(\translate{e_1}_{o_1}\mid\dots\mid\translate{e_n}_{o_n}\mid
                                   \pirecv{o_1}{v_1}.\dots.\pirecv{o_n}{v_n}.\nu h.({!}\pisend{h\cdot\tup}{v_1,\dots,v_n}\mid\pisend
                                 oh))
  \end{align*}

\subsection{Representing arrays}

This section will cover how arrays can be represented in \epic, and how
this is used to translate \butf arrays.
This approach represents each array element with a independent cell,
which users communicate with. Here the extensions in \epic are very
useful, because they allow addressing individual array cells, or all at
once.

\subsubsection{Arrays}

We have decided to represent arrays as a replicated process listening
on some handle, much like how functions are represented in the
\pic. An array element is described by a cell process that listens on
a \textit{broadcast} for a request for all elements and listens on the composed name $\mit{handle}\cdot index$ for a request for a specific element.
\[  
    \mit{Cell}(\mit{handle},\mit{index},\mit{value}) = {!}\pirecv
                                                       {\mit{handle}\cdot\texttt{all}}r.\pisend
                                                       r{\mit{index,
                                                       value}}\mid{!}\pisend{\mit{handle}\cdot\mit{index}}{\mit{index},
                                                       \mit{value}}
                                                       \]
An array is a parallel composition of cells together with a single
replicated sender that provides users of the array with its
length. This is accessed via $h\cdot\texttt{len}$.  Notice how the
different composed labels and numbers, direct messages towards
different listeners in the array.
\begin{align*}
\translate{[e_1,\dots,e_n]}_o&=\nu o_1. \dots. \nu o_n.\nu
    h.(\\
    &\qquad\prod^{n}_{i=1} \translate{e_i}_{o_i}\mid
                                  \pirecv{o_1}{v_1}.\dots.\pirecv{o_n}{v_n}.(\\
    &\qquad\qquad\prod^{n}_{i=1}\mathit{Cell}(h,i-1,v_i)\mid{!}\pisend {h\cdot\texttt{len}}{n}\mid\pisend o{h}))
  \end{align*}
Also, notice how all subexpressions must return a value on their out
channels, before the translation creates the array and returns its
handle.

Indexing is translated similarly to application, however here we compose
the array handle $h$ of the first expression with the index of the
second expression to request the result. 
The check $[i\geq 0]$ is added to make it clear, that the program terminates if an attempt is made to index on a non-positive number.   
\begin{align*}
  \translate{e_1[e_2]}_o=  \nu o_1.\nu
  o_2.(\translate{e_1}_{o_1}\mid\translate{e_2}_{o_2}\mid \pirecv{o_1}h.\pirecv{o_2}i.\bullet[i\geq0]\pirecv{h\cdot i}{i,v}.\pisend ov,\pizero)  
\end{align*}

\subsubsection{Array Operators}

The translation of the \size operator is simple, as the size of an array is sent on the handle channel by the array.
  \begin{align*}
    \translate{\size\ e_1}_o&=\nu o_1.(\translate{e_1}_{o_1}\mid\pirecv{o_1}h.\pirecv {h\cdot\texttt{len}}{n}.\pisend on)\\
  \end{align*}
  \vspace{-1cm}
%
%\newcommand\pitt[3]{\left[#1\right]\:#2
%\tl_if_empty:nTF {#3} {\?} {,#3}
%}
%
%\begin{equation}\label{eq:translate_iota_process_identifier}
%  \begin{aligned}
%    Iota(m,n,r) = \pitt{m = n}{\pisend r {n,n}}\\ 
%    &\quad Iota(m, n/2, r) \mid Iota (n/2+1,n,r)
%  \end{aligned}
%\end{equation}

% \begin{equation}\label{eq:translate_iota}
%     \begin{aligned}
%         &\translate{\mit{iota}\{x\}}_o 
%     \end{aligned}
%         =
%     \begin{aligned}
%         &\nu o_1.(\translate{x}_{o_1} \mid \nu h.\nu\mit{write}.\nu \mit{done}.\pirecv{o_1}{n}.(\\
%         & \qquad Array(h,\mit{write},n, \mit{done})\mid \\ 
%         &\qquad Iota(0,n,\mit{write})\mid\\
%         &\qquad\pirecv{done}{}.\pisend oh\\
%         &\quad))
%     \end{aligned}
% \end{equation}

% \begin{equation}\label{eq:repeat}
%   \mit{Repeat}(n, m, \mit{return}) = [n=m] \pisend{\mit{return}}{n},(\mit{Repeat}(n,n+(m-n)/2)\mid\mit{Repeat}(n+1+(m-n)/2,m))
% \end{equation}

In the translation of the  \iot function below, a process
\textit{Repeat} is created to send numbers $0$ to $n-1$ on the return
channel $r$ (in reverse, but that is not important). Once all numbers
are sent it sends an empty message on $d$ to signal this. \iot then
creates an array in much the same way as usual, but by using the
\textit{Repeat} process instead. Notice how we wait for the done
signal by \textit{Repeat}, before we return the result on $o$, thus
ensuring the call-by-value semantics of \butf.
  \begin{align*}
    \mit{Repeat}(s, r, d)&=\\
    \nu{c}.({!}\pirecv{c}{n}.&[n\geq0] (\pisend{r}{n-1,n-1} \mid
    \pisend{c}{n-1}), \pisend{d}{} \mid \pisend{c}{s})\\
    \translate{\iot\ e_1}_o&=\nu o_1.\nu r.\nu h.(\translate{e_1}_{o_1}\mid\\
                                    &\phantom{=}\qquad\pirecv{o_1}n.\mit{Repeat}(n,r,d)\mid{!}\pirecv r{i,v}.\mit{Cell}(h,i,v)\mid\\
                                    &\phantom{=}\qquad\pirecv d{}.({!}\pisend {h\cdot\texttt{len}} {n})\mid\pisend oh)
  \end{align*}
  A translation of \map must extract the array values from the input
  array and then apply some given function to all these values,
  before they are added back to a new array. A function and the
  \textit{arr} handle are extracted from the input tuple.  The channel
  \textit{vals} is set up such that all values on the array are sent
  on it, followed by a replicated read on all the values.  Each
  element of the output array is initialized after receiving a signal
  on the count channel. This ensures that the done signal is only
  communicated after each array \textit{Cell} has been initialized.
  Once the \textit{done} signal has been communicated, the output of
  the new array handle can be sent on $o$.  This ensures the call by
  value nature of \butf.  Finally, to ensure that \textit{func} is a
  function handle, we invoke it without ever reading the
  result. Otherwise the translation would allow a non-function value
  when the array is empty.

  \begin{align*}
    &\translate{\map\ e_1}_o=\nu{o_1}.\nu h'.( \translate{e_1}_{o_1}\mid\pirecv{o_1}{\mit{args}}.\\
    &\quad\pirecv{\mit{args}\cdot
    \texttt{tup}}{\mit{func},\mit{h}}.\pirecv{\mit{h}\cdot\texttt{len}}{n}.\nu{\mit{vals}}.\pibroad {h\cdot\texttt{all}}{\mit{vals}}.\\
        &\quad\nu{\mit{count}}.(\\
       &\qquad\mit{Repeat}(n,\mit{count},\mit{done})\mid\\
       &\qquad{!}\pirecv{\mit{vals}}{\mit{index},\mit{value}}.\nu r.\pisend{\mit{func}}{\mit{value},r}.\\
       &\qquad\quad\pirecv rv.\pirecv{\mit{count}}{\_, \_}.\mit{Cell}(h',\mit{index},v)\mid\\
       &\qquad\nu o'.\pisend{\mit{func}}{0,
         o'}.\bullet\pirecv{\mit{done}}{}.\pisend
         o{h'}\mid{!}\pisend{h'\cdot\texttt{len}}{n} ))
  \end{align*}

\section{Correctness Criteria}

To be able to analyze the complexity and thus allowing us to reason about the translation, an annotated step notation is introduced. This is inspired by the tick-notation used in \cite{TypesForComplexityOfParallelComputationInPiCalculus}. Here, the $\bullet$ notation marks the important transitions in \epic that match a transition in \butf. 
% To denote transitions leading for example leading up to important transitions we talk about administrative transistions denoted by $\circ$.
% s \in both if needed.

%\begin{figure}
%    \begin{equation*}
%        \msc{Adm}\ \ \begin{prooftree}
%            \hypo{A\mra A'}
%            \infer1{A\mraw A'}
%        \end{prooftree}\qquad
%        \msc{Imp}\ \ \begin{prooftree}
%            \hypo{A \mraw A'}
%            \infer1{\bullet{}A \xrightarrow\bullet A'}
%        \end{prooftree}
%    \end{equation*}
%\label{fig:pi_reduction_rules_labeled}
%\end{figure}
\subsection{Well-Behavedness and Substitution}

In the translation we consider four different kinds of channels:
outputs ($o \in \Omega$), handles ($h \in \Lambda$), signals ($d \in
\Delta$), and collections ($c \in \Psi$).

% For example a handle is often a channel which is used to communicate with an array, tuple, or abstraction, while an output used to send the output of a translation.
%A signal, $d$, can be used to signal an event, by sending nothing on the channel. While a collection is multiple {\color{red}\emph{items}} sent on a single channel.

In the following, we define $U$ as building blocks for translated processes, use
$\mathcal{U}$ as the set of all possible $U$. The intention is that for any $e$
there should exist a process $P$ and $o$ such that $\translate{e}_o\equiv P
\land P\in \mathcal U$.  We define the formation rules for $U$ as follows.
%\begin{equation}
%    \begin{aligned}
%    U ::=\ 
%    &\pirecv ov.U\mid\pirecv{h\cdot\texttt{tup}}{v_1,\dots}.U\mid\pirecv{h\cdot n}{n,v}.U\mid\\
%    &\pibroad{h\cdot\texttt{all}}c.U\mid\pirecv{h\cdot\texttt{len}}n.U\mid\pirecv h{v,o}.U\mid\\
%    &\pirecv{h\cdot\texttt{all}}c.U\mid[n\geq 0]U,0\mid\pirecv c{n,v}.U\mid\pirecv d{}.U\mid\\
%    &{!}\pirecv h{v,o}.U\mid{!}\pirecv{h\cdot\texttt{all}}c.U\mid{!}\pirecv c{n,v}.U\mid\\
%    &[v\neq 0]U,U\mid U|U\mid\nu a.U\\
%    &\pisend ov\mid\pisend h{v,o}\mid \pisend{h\cdot n}{n, v}\mid\pisend{h\cdot\texttt{len}}n\mid\\
%    &\pisend{h\cdot\texttt{tup}}{v_1,\dots}\mid\pisend c{n,v}\mid\mit{Repeat}(n,c,d)\mid\pisend d{}\mid\\
%    &{!}\pisend{h\cdot n}{n,v}\mid{!}\pisend{h\cdot\texttt{len}}n\mid{!}\pisend{h\cdot\texttt{tup}}{v_1,\dots}\mid 0
%    \end{aligned}
%\end{equation}
\begin{equation}
    \begin{aligned}
        U ::=\ 
        %%%%%%%%%%%%%%%%%%%%%%%%%%%%%%%%%%%%%
        %%%%%%%%%% Non Terminating %%%%%%%%%%
        %%%%%%%%%%%%%%%%%%%%%%%%%%%%%%%%%%%%%
        % 1
        &\pirecv ov.U\mid
        % 2
        \pirecv h{v,o}.U\mid
        {!}\pirecv h{v,o}.U\mid
        % 3
        \pirecv{h\cdot n}{n,v}.U\mid
        % 4
        \\&\pirecv{h\cdot\texttt{len}}n.U\mid 
        % 5
        \pirecv{h\cdot\texttt{tup}}{v_1,\dots}.U\mid
        % 6
        \pirecv{h\cdot\texttt{all}}c.U\mid
        {!}\pirecv{h\cdot\texttt{all}}c.U\mid
        \pibroad{h\cdot\texttt{all}}c.U\mid
        % 7
        \\&\pirecv c{n,v}.U\mid
        {!}\pirecv c{n,v}.U\mid
        % 8
        \pirecv d{}.U\mid
        % 9
        [n\geq 0]U,0\mid
        [v\neq 0]U,U\mid
        % ETC
        U|U\mid
        \nu a.U\mid
        %%%%%%%%%%%%%%%%%%%%%%%%%%%%%%%%%
        %%%%%%%%%% Terminating %%%%%%%%%%
        %%%%%%%%%%%%%%%%%%%%%%%%%%%%%%%%%
        % 1
        \\&\pisend ov\mid
        % 2
        \pisend h{v,o}\mid
        % 3
        \pisend{h\cdot n}{n, v}\mid
        {!}\pisend{h\cdot n}{n,v}\mid
        % 4
        \\&\pisend{h\cdot\texttt{len}}n\mid
        {!}\pisend{h\cdot\texttt{len}}n\mid
        % 5
        \pisend{h\cdot\texttt{tup}}{v_1,\dots}\mid
        {!}\pisend{h\cdot\texttt{tup}}{v_1,\dots}\mid
        % 7
        \\&\pisend c{n,v}\mid
        % 8
        \pisend d{}\mid
        % ETC
        \mit{Repeat}(n,c,d)\mid
        0
    \end{aligned}
\end{equation}
  Here, we consider $o\in\Omega$, $h\in\Lambda$, $d\in\Delta$, $c\in\Psi$, and $a\in\Omega\cup\Lambda\cup\Delta\cup\Psi$.
  The terms $v,v_1,v_2,\dots$ are used to signify numbers $n$ or handles, and we use $\Theta$ for these.
  Therefore for channels $o$ and $h$, it holds that $\pisend oh\in U$, and $\pisend o5\in U$, while $\pisend oo\not\in U$.
  
  Note that we use members of the sets above in name binding also, which is to signify which ``type'' of term is expected to be received on the channel.
  For example in $\pirecv ov.U$ for $o\in\Omega$ and $v\in\Theta$, the variable $v$ might be present in the process $U$ where it is used as a value assuming that what is sent on $o$ is actually in $\Theta$.
  If the term $t$ received on $o$ is not in $\Theta$, $U\pisubst vt$ would also not be in $\mathcal U$.
  % Dette er meget ligesom den paper jeg har læst tidligere

 \cref{lmm:stay_alive} ensures that any process
  $U\in\mathcal U$ continues to be well-behaved in regards to the
  translation channels.
  
\begin{lemma}\label{lmm:stay_alive}
    For any process $U$ it holds that if $U\mra P'$ and $U$ is not
    observable on any channel, then $P'\in\mathcal U$.
\end{lemma}

In \butf function application is done by substituting a value into the function body.
For numbers, this is simple as the number simply gets substituted into the process.
However, in the translated process, the function, array, and tuple servers cannot be substituted into a process, and therefore, lie outside of it.
This creates a structural difference between $\translate e_o$ and $\translate{e'}_o$, which \Cref{lmm:find_replace} shows that they still behave the same under $\ba$.
\begin{theorem}
    For values $e_1$ and arbitrary expressions $e_2$, we have that
    \begin{enumerate}
        \item if $e_1$ is a number ($n$) then
          $\transtopi{e_2}_o\pisubst xn\ba\transtopi{e_2\{x\coloneqq
          n\}}_o$ for some $o$,
        \item or if $e_1$ is an abstraction, tuple, or array then $\nu
          h.(Q\mid\translate{e_2}_o\pisubst
          xh)\ba\translate{e_2\{x\coloneqq e_1\}}_o$ for some $o$.
        Here, $Q$ is $\translate{e_1}_o$ after sending $h$ on $o$, i.e. $\translate{e_1}_o\mid\pirecv ox.P\xRightarrow{\bullet}\nu h.(Q\mid P\pisubst xh)$.
    \end{enumerate}
     \label{lmm:find_replace}
\end{theorem}
\begin{proof}
    \begin{enumerate}
    \item When $\translate{n} = \pisend{o}{n}$ and $\translate{x} =
      \pisend{o}{x}$, we have that $ {\pisend{o}{x}}\pisubst xn \ba
      {\pisend{o}{x}} \{x\coloneqq n\} \land \pisend{o}{x}\{x\coloneqq
      n\} = \pisend{o}{n} $.  Thus $\pisend{o}{x}\pisubst xn \ba \pisend{o}{n}$.

    \item  In the process $\translate{e_2\{x\mapsto e_1\}}_o$ there
      can be different servers all of which stem from the translation of $e_1$.
    Each of these servers can have a number of usages, where a handle is communicated on to access a specific server.
    We denote $\mathcal P=\{P_1,\dots, P_n\}$ as a collection of usages
    of these servers in the translation, and therefore
    $P\subseteq\mathcal U$. And $\mathcal Q=\{Q_1,\dots,Q_m\}$ is a collection of servers, such that for some $Q_i\in\mathcal Q$, $Q_i$ communicates on $h_i$ instead of $h$.

    Now we introduce the relation $R$, which relates processes with a single server channel with processes where the same server channels is repeated for multiple handle channels.
    Here processes in $\mathcal P$ are in a context and either communicate with a single $Q$ on $h$ (the left side), or with multiple $Q$'s with multiple $h$'s (the right side).
    The function $f : \mathcal Q\mra \mathbb P(\mathcal P)$ takes a single $Q_i$ and returns the uses of said $Q$, these uses $P_i$ normally use the channel $h$, but have to be substituted to use $h_i$.
    \begin{equation}\label{eq:conject_thing_2}
    \begin{aligned}
        R = \{&(K[\nu h.\nu\mathcal A.(Q\mid\prod_{P_i\in\mathcal P}P_i\mid U)],K[\nu h_1.\dots\nu h_m.\nu\mathcal A.\\
        &(\prod_{Q_i\in\mathcal Q} Q_i\mid \prod_{Q_l\in\mathcal Q}\prod_{P_i\in f(Q_l)} P_i\pisubst h{h_l}\mid U)])\mid\\
        &U\not\downarrow_h\land \forall i\in[1..n].U\downarrow_{h_i}\land\bigcup_{Q_i\in\mathcal Q}f(Q_i)=\mathcal P\\
        &\land
    \forall Q_i,Q_j\in\mathcal Q.f(Q_i)\cap
    f(Q_j)=\varnothing\land\forall a.((\nu\mathcal A.U)\not\downarrow
    a)\\
        &\}
    \end{aligned}
    \end{equation}
    Now we show that $R$ is a WABB, by considering the transitions each side can take.
    First, consider when the left transitions, and identify four cases.
    \begin{enumerate}
        \item For an internal communication of the form $K[0]\mra
          K'[0]$, we can use the same $K'$ on the right side, and show
          that the new pair is in $R$. 
        \item For an internal communication in $U$ of the form $U\mra
          U'$ we might introduce a new process $P$ or $Q$ in $\mathcal
          U$, which can be moved out of $U'$ such that
          $U'\not\downarrow_h$.  We can match this transition on the
          right side, and through $\equiv$ the pair is still in $R$. 
        \item A $P_i$ communicates with $Q$ on the channel $h$.
        \[
            P_i\mid Q\mra Q\mid S
        \]
        We consider the different forms which $Q$ an take, depending on whether $e_1$ is an abstraction, tuple, or array.
        \begin{enumerate}
            \item If $e_1=\lambda x.e_b$, then $Q$ takes the form shown below.
            \[
                Q={!}\pirecv h{x,r}.\translate{e_b}_r
            \]
            Given that $P_i$ communicated with $Q$, means that $P_i=\pisend h{x,r}.S'$, where $S'$ is some arbitrary process.
            This communication will therefore uncover $S$ and spawn $\translate{e_b}_r$.
            Similarly to the first case, we can find a new $U'$, and
            $\mathcal P'$, such that $U'$ does not contain $h$. Then
            the following holds. 
            \[
                Q\mid \prod_{P_i\in \mathcal P}P_i\mid U\mra Q \mid
                \prod_{P_i\in\mathcal P'}P_i\mid U'
            \]
            With the right side, the same transition can be taken by $P_i\pisubst h{h_l}$ on the channel $h_l$ with the server $Q_l$.
            Here, we can find a new $\mathcal Q'$ such that the pair resulting from the two transitions is in $R$.
            \item If $e_1=(e_{1,1},\dots,e_{1,o})$, then $Q$ is as follows for some $T_1$ to $T_o$
            \[
            Q = {!}\pisend{h\cdot(-1)}{T_1,\dots,T_o}
            \]
            The proof proceeds as in the case of abstraction, except
            that now $Q$ will not spawn any new processes.
            \item If $e_1=[e_{1,1},\dots,e_{1,o}]$.
            Here, $Q$ will be as follows for terms $T_1$ to $T_o$.
            \[
            Q = \prod_{i\in 1..o}({!}\pirecv {h\cdot\texttt{all}}r.\pisend r{T_i}\mid{!}\pisend{h\cdot i}{T_i})\mid {!}\pisend h{o}
            \]
            Because $Q$ is a collection of parallel replicated sends and receives, it acts in much the same way as in the case of tuples.
            We can therefore follow the same reasoning as in the
            previous cases.
        \end{enumerate}
        \end{enumerate}
        We know that no other cases exists for the transition, given that $Q$ and $P_i$ cannot communicate with either $U$ or $K$, given that these processes do not contain $h$.

        We now consider then the right side of a pair in $R$ transitions, and again identify four cases
        \begin{enumerate}
            \item Internal communication in K, which is similar to case (2) above.
            \item Internal communication in U, which is similar to case (3) above.
            \item A $P_i\pisubst h{h_j}$ communicates with a $Q_j$ on a channel $h_m$.
            \[
                P_i\pisubst h{h_j} | Q_j \mra Q\mid S
            \]
            We only consider the case when $e_1= \lambda x.e_b$, as the other cases easily follow.

            In this case $Q_j$ will again take the form shown below.
            \[
                Q_j = {!}\pirecv {h_j}{x,r}.\translate{e_b}_r
            \]
            Like in case (4) above, we can construct new $U'$ and $\mathcal P'$ to accommodate the new processes after the reduction.
        \end{enumerate}

        Finally, we must show that the pair below is in $R$.
        \[
          (\nu h.(Q\mid\translate {e_2}_o\pisubst xh),\translate{e_2\{x\coloneqq e_1\}}_o)
        \]

        % In $e_2$ a number of uses of the variable $x$ exists. In the
        % translation $\translate{e_2}_o\pisubst xh$, each of these
        % usages of $x$ have been replaced by $\pisend{o'}x\in U$. VED
        % IKKE HVORFOR VI HAR PISUBST MED HER.
        In $e_2$ a number of uses of the variable $x$ exists. In the
        translation $\translate{e_2}_o$, each of these usages of $x$
        have been replaced by $\pisend{o'}x\in U$.
        In $e_2\{x\coloneqq e_1\}$ each of the $x$'es have been replaced by the whole of $e_1$, and the translation $\translate {e_2\{x\coloneqq e_1\}}_o$ then contains multiple instances of $\translate{e_1}_{o'}$ for some output channel $o'$.
        Each of these instances has the form $\translate{e_1}_{o'} = \nu h.(Q\mid\pisend{o'}h)$.

        We know that both $\translate{e_2}_o\pisubst xh$ and
        $\translate{e_2\{x\mapsto e_1\}}_o$ are in $\mathcal U$, and we
        can match them to a pair in $R$ by structural congruence.

    \end{enumerate}
\end{proof}

\subsection{Operational Correspondence}

We consider the translation to be correct when it preserves the
reduction sequence and the result of the program.  To do this, we
define an operational correspondence, which ensures translation
correctness.

\begin{definition}[Administrative Operational Correspondence]
Let $R$ be a binary relation between an expression and a process.
Then $R$ is an \emph{administrative operational correspondence} if $\forall (e,P) \in R$ it holds that
    \begin{enumerate}
        \item if $e \mra e'$ then there $\exists P'$ such that $P \xRightarrow\bullet\ba P'$ and $(e',P')\in R$, and
        \item if $P \xRightarrow\bullet P'$ then there $\exists e',Q$ such that $e\mra e'$, $Q\ba P'$, and $(e',Q)\in R$.
    \end{enumerate}
    We denote $e \opcor P$ if there exists an operational correspondence relation $R$ such that $(e,P) \in R$.
\end{definition}

This definition achieves soundness by guaranteeing that all reductions
that happen in a \butf program $e$ can be matched by a sequence of
reductions in the corresponding \epic process $P$. The completeness is
ensured by requiring that for any important reduction
$P \xRightarrow\bullet P'$ where $e \opcor P$, we have that
$e$ can evolve to some $e'$ for which there exists some $Q$ where
$e' \opcor Q$ and $Q$ is bisimilar to $P'$.

We will now attempt to prove administrative operational correspondence
for \butf and \epic. The lemma below is used to identify the possible
reduction cases when $P$ is contained within a context, and is
usefully for simplifying program behavior.
\begin{lemma}
    For any $P$ and $C$, if $Q$ exists such that $C[P]\xrightarrow s Q$, then one of the following holds:
    \begin{enumerate}
        \item $C$ reduces alone, thus $Q=C'[P]$ with context $C'$ such that $C[\pizero]\xrightarrow sC'[\pizero]$,
        \item $P$ reduces alone, thus $Q=C[P']$ with $P\xrightarrow sP'$, and
        \item $C$ and $P$ interact, thus $Q=C'[P']$ for $P'$ and $C'$ such that $O$ exists where $O\mid P\xrightarrow sO'\mid P'$, $C[P]\xrightarrow sC'[P']$, and $C[P]\equiv\nu{\Vec a}.(O\mid P)$.
    \end{enumerate}
    \label{lmm:process_in_context}
\end{lemma}

The first step to prove the operational correspondence, is proving 
that values always send on $o$. This is shown in the following lemma.
\begin{lemma}
  Let $e$ be a value. Then
  $\exists P.\transtopi e_o\xrightarrow{\circ} P\wedge
  P\downarrow_{\overline o}$.
    \label{lmm:to_o_or_not_to_o}
\end{lemma}
\begin{proof}
  We let $\mathcal D(e)$ denote the depth of $e$
  and proceed by induction on $D(e)$. If $e$ is a number or an
  abstraction, then $\mathcal D(e)=0$.  However, if $e$ is a tuple or
  array with elements $e_0$ to $e_{m}$, then
  $\mathcal D(e)=\max_{i\in[0..m]}(\mathcal D(e_i)) + 1$.  By
  induction on $\mathcal D(e)$ we show that the lemma holds for all
  $e$.  In the base case $\mathcal D(e)=0$, and thus $e$ is either a
  number or abstraction.  From the translation of a number or an
  abstraction, we know that $\transtopi e_o\downarrow_{\overline o}$,
  which is consistent with the lemma for $e$.  In the inductive case,
  where $\mathcal D(e)>0$, $e$ must be either a tuple or array with
  elements $e_0$ to $e_m$.  Here, the lemma holds for all $e'$ where
  $\mathcal D(e')<\mathcal D(e)$, and in extension $e_0$ to $e_n$.  If
  $e$ is a tuple, then we can take reductions such that
  $\transtopi{e_0}_{o_0}$ to $\transtopi{e_m}_{o_m}$, all send on
  channels $o_0,\dots,o_m$.  Then $\pisend oh$ is unguarded.  For
  array, after it has sent on channels $o_0$ to $o_m$, it can then
  receive on \textit{done} because it has sent $m+1$ values.  Then
  $\pisend o{\mit{handle}}$ is unguarded.
\end{proof}

The proofs of the next two lemmas can be found in
\cite{p10}. The first lemma is used to remove the no longer used parts
of the program and thus allows for a simple garbage collection.
\begin{lemma}
    If $P\ba 0$, then for all $Q$ it holds $P\mid Q\ba Q$.
    \label{lmm:garbage_collection}
\end{lemma}

The second lemma is the converse of \cref{lmm:to_o_or_not_to_o}. It tells us that if the encoding of a \butf expression
$e$ is eventually able to output on the $o$ name, then $e$ is a value.

\begin{lemma}
    If for some expression $e$, $\exists P.\transtopi e_o\xRightarrow{\circ} P\wedge P\downarrow_{\overline o}$ then $e \in \mathcal{V}$.
    \label{lmm:to_o_or_not_to_o_part_two}
  \end{lemma}
  
We now construct an administrative operational correspondence whose pairs
consist of \butf programs and their corresponding translations.

\begin{theorem}
  For any \butf program $e$ and fresh name $o$ we have that
  $e \opcor \transtopi{e}_o$.
    \label{lmm:operationel_correspondence}
\end{theorem}
\begin{proof}
  Let $\mathcal B$ be the set of all \butf programs and let $R$ be the
  relation $R=\{(e,\transtopi{e}_o) \mid e\in \mathcal{B}, o \; \text{fresh}\}$. We show
  that $R$ is an administrative operational correspondence.
    
    We only consider pairs where $e\mra e'$ and where $\transtopi e_o$ contains $\bullet$.
    By extension of this, we are not considering values.

\smallskip\noindent\textbf{Array}\ $e= [e_1,\dots,e_n]$\quad
Let us first consider that $e\mra e'$, and from the \butf semantics, we know that there must exist an $i$ such that $e_i\mra e_i'$.
% Therefore we show that $\translate e_o\xRightarrow{\circ}\xrightarrow\bullet P$ such that $P\ba\translate{e'}_o$, in that $(e',\translate{e'}_o)\in R$.
Here, $\translate{e}_{o}$ contains $\nu o_i.( \translate{e_i}_{o_i})$.
We assume that $(e_i, \translate{e_i}_o)\in R$, and therefore we know that $\translate{e_i}_{o_i}\xRightarrow{\circ}\xrightarrow\bullet Q$ such that $Q\ba\translate{e_i'}_{o_i}$.
Let $P$ be the process $\translate e_o$ with the subprocess $\translate{e_i}_{o_i}$ replaced by $Q$.
In that $\translate{e_i}_{o_i}$ is unguarded in $\translate e_o$, we know that $\translate e_o\xRightarrow{\circ}\xrightarrow\bullet P$, and that $P\ba\translate{e'}_o$.

Vice versa, we show that if $\translate e_o\xRightarrow{\circ}\xrightarrow\bullet P$, then $P\ba\translate{e'}_o$ where $e\mra e'$.
Given that the translation of array does not have $\bullet$, the important reduction must happen inside $\translate{e_i}_{o_i}$.
The translation ensures that $\translate{e_i}_{o_i}$ can only be observed on $o_i$, and therefore the different $\translate{e_i}_{o_i}$ cannot reduce together.
We know that there exists a $j$, such that the important reduction occurs in $\translate{e_j}_{o_j}$, and because $(e_j, \translate{e_j}_o)\in R$, we know that $\translate{e_j}_{o_j}\xRightarrow{\circ}\xrightarrow\bullet Q$ for some $Q$ and $e_j'$ where $Q\ba\translate{e_j'}_{o_j}$ and $e_j\mra e_j'$.
We can then select $e'$ as $e$ where $e_j$ has been replaced by $e_j'$, and then $P\ba\translate{e'}_o$.
This is because $P$ and $\translate{e'}_o$ only differ by administrative reduction (for example in some other $\translate{e_i}_{o_i}$ for $i\neq j$).

% Vice versa if for $P$ and $P'$ where $\transtopi e_o\xRightarrow{\circ} P\xRightarrow{\circ}\xrightarrow\bullet P'$ we must show that there exists an $e'$ such that $e\mra e'$ and $P'\ba\transtopi{e'}_o$.
% From \cref{lmm:to_o_or_not_to_o} and the construction of $\transtopi{}_o$ we know that $\transtopi e_o$ cannot be observed by $\overline o$ until it has taken an important reduction.
% We therefore rearrange reductions between $\transtopi e_o$ and $P'$ such that administrative reductions contained in $\transtopi{e_i}_{o_i}$ happen after $P$, and all other before.
% Process $P$ can then be restated as a context, i.e. $P=C[\transtopi{e_i}_{o_i}]$.
% In that $(e_i,\transtopi{e_i}_{o_i})\in R$, for any $Q$ where $\transtopi{e_i}_{o_i}\xRightarrow{\circ}\xrightarrow\bullet Q$ then $Q\ba\transtopi{e_i'}_{o_i}$.
% We then know that $P'$ can be expressed as $C[Q]$ where $\transtopi{e_i}_{o_i}\xRightarrow{\circ}\xrightarrow\bullet Q$, and thus $C[Q]=P'\ba\transtopi{e'}_o$.

\smallskip\noindent\textbf{Tuple}\ $e=(e_1,\dots,e_n)$\quad
Follows from the same argument as \textbf{Array}.
%Indexing

\smallskip\noindent\textbf{Indexing}\ $e=e_1[e_2]$\quad
Operational correspondence requires that if $e\mra e'$ then $\transtopi{e}_o\allowbreak\xRightarrow\bullet\allowbreak P$ such that $P\ba\transtopi{e'}$.
The translation of indexing is defined as seen below.
\begin{equation}
\begin{aligned}
\transtopi{e_1[e_2]} =\ &\nu o_1.\nu o_2.(\transtopi{e_1}_{o_1} \mid \transtopi{e_2}_{o_2}\\
                &\mid\pirecv{o_1}{h}.\pirecv{o_2}{i}.\bullet[i\geq 0]\ \pirecv{h \cdot i}{i, v}.\pisend{o}{v},\pizero)
                \end{aligned}
                \end{equation}
There are three rules for indexing in \butf which we shall call (\textsc{E-Index}, \textsc{E-Index-1}, and \textsc{E-Index-2}).
On the other hand, in \epic, there is the translation for the array ($\transtopi{e_1}_{o_1}$) and the expression to define the desired index ($\transtopi{e_2}_{o_2}$).
The \textsc{E-Index-1/2} rules are used to evaluate sub-expressions $e_1$ and $e_2$.
Because $(e_1,\transtopi{e_1}_0)\in R$, if $e_1[e_2] \mra e_1'[e_2]$ then $\transtopi{e_1}_{o_1} \xRightarrow\bullet\ba\transtopi{e_1'}_{o_1}$.
Then because $\transtopi{e_1}_{o_1}$ is unguarded in $\transtopi{e}_o$, \cref{eq:indexing_stuff_this_is_unique} holds.
\begin{equation}\label{eq:indexing_stuff_this_is_unique}
\begin{aligned}
    \transtopi{e}_o\xRightarrow\bullet\ba\ &\nu{o_1}.\nu{o_2}.(\transtopi{e_1'}_{o_1}\mid\transtopi{e_2}_{o_2}\mid\\
    &\pirecv{o_1}h.\pirecv{o_2}i.\bullet[i\geq 0]\ \pirecv{h\cdot i}v.\pisend ov,\pizero)=\transtopi{e'}_o
    \end{aligned}
\end{equation}
The same has to hold for $e_2$.
These must be assumed to hold if all other cases are operationally correspondent since $e_1$ and $e_2$ are in $R$.

The actual indexing operation (\textsc{E-Index}) is also relevant here.
Here, we know that if $v_1[v_2] \mra v_3$ then we have to have the corresponding operation $\transtopi{v_1[v_2]}_o \xRightarrow{\bullet} \ba\transtopi{v_3}_o$.
Because $e\mra e'$ by \textsc{E-Index}, we know that $e_1$ is an array of length $m$ and $e_2$ is an integer less than $m$.
With the translation $\nu o_1.\nu o_2.(\transtopi{e_1}_{o_1} \mid \transtopi{e_2}_{o_2}
\mid\pirecv{o_1}{h}.\pirecv{o_2}{i}.\bullet[i\geq 0]\ \pirecv{h \cdot i}{v}.\pisend{o}{v})$ we know that they are ready to send on their $o$ after some administrative reductions channels by \cref{lmm:to_o_or_not_to_o}.
The translation thus proceeds to send the handle of the array via $o_1$ and the value is sent on $o_2$. These are administrative reduction and are thus covered by the $\xRightarrow{\circ}$ reductions.

This reduces the program down to $\nu h.(Q_h\mid \bullet\piifthen{i > 0}{\pirecv{h \cdot i}{v}.\pisend{o}{v})}$, where $Q_h$ is the leftovers from the array $\transtopi{e_1}_{o_1}$ and $i$ is the index from $e_2$.
Next, we have the if statement together with $\bullet$, which is defined as an important reduction, and is expressed by the $\xrightarrow\bullet$ arrow: $\dots\xrightarrow\bullet\nu o_1.\nu o_2.(Q_h\mid \pirecv{h \cdot i}{v}.\pisend{o}{v'})$.
Lastly, the value is received internally as $v'$ and returned along the out-channel ($o$).
The still existing array $Q_h$ can now be garbage collected by \cref{lmm:garbage_collection}.

We must also show that if $\transtopi e_o\xRightarrow{\circ}\xrightarrow\bullet P$ then we can find $e'$ such that $P\ba\transtopi{e'}_o$ and $e\mra e'$.
The important reduction can either happen inside either $\transtopi{e_1}_{o_1}$ or $\transtopi{e_2}_{o_2}$ (very similar to array), or before the index check.
In the first case, we can find a $e'$ much like in arrays.
In the latter case, we know that $e_2$ and $i$ are integers that are greater or equal to zero and that some process is listening on $h\cdot i$.
This is only the case if $e_2$ is an array of size larger than $i$.
With this, we know that $e\mra$ by \textsc{E-Index}.
% COMMENT
%The still existing array can now be garbage collected in this case since its out-channel is restricted to itself and is therefore weak administrate barbed bisimilar to the empty process ($\nu{o_1}.\transtopi{e_1}_{o_1}\ba\piz$).
%BinOp - unimportant

%Maybe (probably not) Abstraction val/tup

%Application
\smallskip\noindent\textbf{Application}\ $e \coloneqq e_1 \  e_2$\quad
There are two cases for which $e \mra e'$. 
One case is when the subexpressions $e_1$ or $e_2$ can reduce.
In that $\transtopi{e_1}_{o_1}$ and $\transtopi{e_2}_{o_2}$ appear unguarded in $\transtopi{e}_o$ and since $\{(e_1,\transtopi{e_1}_o),(e_2,\transtopi{e_2}_o)\} \subseteq R$, we know that $\transtopi{e}_o$ can match $\xRightarrow{\circ}\xrightarrow\bullet\ba$.

%One case is applying the \textsc{E-App-1/2} rules on $e_1$  $e_2$ which are .
The second case is when $e_1 \not\mra \wedge \ e_2 \not\mra$.
Here, \textsc{E-Beta} can take an important reduction.
These are matched by the translation of application.
\begin{align*}
&\nu{o_1}.\nu{o_2}.(\transtopi{e_1}_{o_1} \mid \transtopi{e_2}_{o_2} \mid \pirecv{o_1}{f}.\pirecv{o_2}{x}.\bullet\pisend{f}{x,o}) & \xRightarrow{\circ} \\
&\nu{o_1}.\nu{o_2}.(\nu{f'}.\pisend{o_1}{f'}.({!}\pirecv{f'}{x,r}.\transtopi{e_b}_r) \mid\\
&\quad\nu v.(\pisend{o_2}{v} \mid S) \mid \pirecv{o_1}{f}.\pirecv{o_2}{x}.\bullet\pisend{f}{x,o})  &\xRightarrow{\circ}\\
&\nu v.\nu{f'}.({!} \pirecv{f'}{x,r}.\transtopi{e_b}_r \mid \bullet\pisend{f'
}{v,o}) \mid S &\xrightarrow\bullet\\
&\nu{f'}.({!} \pirecv{f'}{x,r}.\transtopi{e_b}_r) \mid \nu v.(F_o \mid S) &\ba\\
&\nu v.(F_o\mid S)
\end{align*}
First, the expressions are evaluated to values such that they are ready to send on the out-channels.
This results in a guarded replicated function server for $e_1$ and a value ready to be sent for $e_2$.
Afterward, the administrative reductions, in the form of communicating along the out-channels, are performed.

We know that $e_1$ is an abstraction, $\lambda x. e_b$, and therefore $\transtopi{e_1}_{o_1}\allowbreak=\allowbreak\nu f\allowbreak.\allowbreak({!}\allowbreak\pirecv f{x,r}\allowbreak.\allowbreak\transtopi{e_b}_r\allowbreak\mid\allowbreak\pisend{o_1}{f})$.
Also note that $S$ is the process needed to maintain value $v$, i.e. $\transtopi{e_2}_{o_2}\ba\nu a.(S\mid\pisend{o_2}v)$ such that $S$ is only observable on $a$ or $\overline a$.

After the two subprocesses have sent their value on $o$, we can send on $f'$ which is marked by a $\bullet$.
By sending $(v,o)$ an instance of $\transtopi{e_b}_r$ is unguarded, where the name of the return channel is substituted with the name of the out-channel ($o$) together with the value ($v$).

%$F$ is left undefined to represent a generic program that is responsible itself of resolving.

We let $F_o$ denote the function body $\transtopi{e_b}_r$ with the return channel $o$ and the value of $\transtopi{e_2}_{o_2}$, ie. $F_o=\transtopi{e_b}_o\pisubst xv$.
$F_o$ corresponds to the translation of $e'=e_b\{x\coloneqq e_1\}$ by \cref{lmm:find_replace}, and thus $\transtopi e_o\xRightarrow{\circ}\xrightarrow\bullet\ba \transtopi{e'}_o$.

If $\transtopi e_o\xRightarrow{\circ}\xrightarrow\bullet P$ then we must show that $e'$ exists such that $P\ba\transtopi{e'}_o$ and $e\mra e'$.
Like with arrays, if $\xrightarrow\bullet$ happens entirely inside either $\transtopi{e_1}_{o_1}$ or $\transtopi{e_2}_{o_2}$ then, we can select $e'=e_1'\ e_2$ or $e'=e_1\ e_2'$.
If $\xrightarrow\bullet$ happens when sending on $f$, then both $\transtopi{e_1}_{o_1}$ and $\transtopi{e_2}_{o_2}$ can send on $o$ after some administrative reductions.
Therefore by \cref{lmm:to_o_or_not_to_o_part_two} $e_1$ and $e_2$ must be values.
Also $\transtopi{e_1}_{o_1}$ must send the name of a function channel on $o_1$ and therefore we know that $e_1=\lambda x. e_b$ or $e_1=\lambda p. e_b$.
Therefore by \textsc{E-Beta} we have $e \to e'$ where $e'=e_b\{p:=e_2\}$.

%If $\transtopi e_o\xRightarrow{\circ}\xrightarrow\bullet P$ then we know that $P\ba\transtopi{e'}_o$ per same argument as with name binding.

% COMMENT

%let val/tup
\smallskip\noindent\textbf{Conditional}\ $e = \texttt{if}\  e_1 \  \texttt{then} \  e_2 \  \texttt{else} \  e_3 $\quad

% There are two rules for the if/else conditionals in \butf, which we name the following way, \textsc{E-If-True} if the conditional is non-zero, and \textsc{E-If-False} if the condition is zero.
The translation for $e$ is as seen below.\[
\nu{o_1}.(\transtopi{e_1}_{o_1}\mid\pirecv{o_1}{v}.[v \neq 0]\ \transtopi{e_2}_{o},\transtopi{e_3}_{o})
\]
We know that any reduction done by $e_1$, can be matched by $\transtopi{e_1}_{o_1}$ since $\transtopi{e_1}_{o_1}$ is unguarded and $(e_1,\transtopi{e_1}_o) \in R$.
Once $e_1$ is done and can send some term ($M$) on $o_1$, there is only one reduction left.
This reduction 
reduces $[M\neq0]\ \transtopi{e_2}_{o},\transtopi{e_3}_{o}$ to either $\transtopi{e_2}_o$ or $\transtopi{e_3}_o$.
Since $e\mra$ and $e_1\not\mra$, $e_1$ must be a value, and thus either \textsc{E-If-True} is matched and \cref{eq:chris_er_sej} or \textsc{E-If-False} is matched and \cref{eq:lars_er_sej}.
\begin{equation}\label{eq:chris_er_sej}
[M\neq0]\ \transtopi{e_2}_{o},\transtopi{e_3}_{o} \xrightarrow\bullet \transtopi{e_2}_o
\end{equation}
\begin{equation}\label{eq:lars_er_sej}
\piifthenelse{M \neq 0}{\transtopi{e_2}_{o}}{\transtopi{e_3}_{o}} \xrightarrow\bullet {\transtopi{e_3}_{o}}
\end{equation}

In the other case when $\transtopi e_o\xRightarrow\bullet P$, we can show that $e'$ exists such that $P\ba e'$ and $e\mra e'$, in much the same way as with name binding.

%\smallskip\noindent\textbf{Tuple application}\ $e = e_1{:}e_2 $\quad{\color{red} mangler}

\smallskip\noindent\textbf{Map}\ $e = \map\ e_1$\quad
First we consider the case where $e\mra e'$.
Like in previous cases, we have can match transitions to the $e_1$ subexpression with the unguarded $\translate{e_1}_{o_1}$ in $\translate e_o$.

This leaves us with the case where $e_1$ is the tuple $((\lambda x.e_b), [v_1,\dots,v_n])$, such that the \map transition can occur.
Then $e'$ becomes the following.
\[
    e' = [e_b\{x\mapsto v_1\}, \dots, e_b\{x\mapsto v_n\}]
\]

We can then see that $\translate{\map\ e_1}_o$ only differs from $\translate{e'}_o$ by some additional administrative reductions.
These happen when the tuple is unpacked, and when each function/substitution is done before the \textit{Cell}.

We follow the same argument to state that $\translate{e}_o\xRightarrow\bullet\ba\translate{e'}_o$.

Additional if $\translate{e}_o\xRightarrow\bullet P'$ then we must be able to find $e'$ such that $P'\ba\translate{e'}_o$.
We know that $P'$ must have taken transition $\bullet\pisend{\mit{done}}{}$, meaning $e_1$ is a tuple value due to $o_1$ and $\mit{args}\cdot\texttt{tup}$ requiring an receive action.
We also know that the tuple must contain an array in the second parameter, and due to the dummy send on \textit{func}, that the first is a function.
Then $e=(\lambda x.e_b, [v_1, \dots, v_n])$ and $e'$ can be set as follows.
\[
  e' = [e_b\{x\mapsto v_1\}, \dots e_b\{x\mapsto v_n\}]
\]

\smallskip\noindent\textbf{Size}\ $e = \texttt{size}\ e_1 $\quad
Follows same argument as the map case.

\smallskip\noindent\textbf{Iota}\ $e = \iot\ e_1 $\quad
Follows same argument as the map case.
\end{proof}

% \begin{proof}
%     To prove this we have to consider every expression for which there exists $P$ such that $\translate e_o\xRightarrow{\circ} P\land P\downarrow_{\overline o}$, and show that these are all values.
%     First for the set of values, this holds trivially as$ \translate x_o = \pisend ox\land
%     \translate n_o = \pisend on$
% The same applies to the translation of functions as $\translate{\lambda x.e}_o =\nu f.(\pisend of\mid{!}\pirecv f{x,r}.\translate{e}_r)$
%  Finally, there is array and tuples containing only values. 
% For value tuples there is only one possible reduction at a time such that.
% \begin{equation*}
% \begin{aligned}
% &\nu{o_1}.\dots.\nu{o_n}.(\pisend{o_1}x_1\mid\dots\mid\pisend{o_n}x_2\mid\pirecv{o_1}{v_1}.\dots.\pirecv{o_n}{v_n}.\\
% &\quad\nu h.({!}\pisend{h\cdot(-1)}{v_1,\dots,v_n}\mid\pisend oh)) \rightarrow^n \\
% &\quad\nu h.({!}\pisend{h\cdot(-1)}{x_1,\dots,x_n}\mid\pisend oh)\end{aligned}
% \end{equation*}

% Value arrays also only have one reduction available at once, and thus
% \begin{align*}
% \nu o_1\dots\nu o_n.\nu
%   h.(\pisend{o_1}{v_1}\mid\dots\mid\pisend{o_n}{v_n}\mid
%   \pirecv{o_1}{x_1}.\dots.\pirecv{o_n}{x_n}.\nu
%   b.(\mathit{Cell}(h,b,1,x_1)\mid\dots\mid\mathit{Cell}(h,b,n,x_n)\mid{!}\pisend
%   h{b,n}\mid\pisend oh)) \\
%   \rightarrow^n \mathit{Cell}(h,b,1,v_1)\mid\dots\mid\mathit{Cell}(h,b,n,v_n)\mid{!}\pisend h{b,n}\mid\pisend oh)
% \end{align*}

% As such we can see that all values fulfills the lemma.
% \end{proof}%!!!!!!

\section{Work and Span Analysis}
To compare the work (W) and span (S) with those of Futhark we carry
out an analysis on the translation of \butf into \epic.
We define work as the actual instructions that happen and span as the depth of parallel instructions.
Our cost model is based on the number of $\bullet$-marked reductions encountered which were placed earlier to facilitate operational correspondence. We find this definition of work useful, but can also see that this definition and the $\bullet$ placements is arbitrary when using it to define work. With this definition, we want to illustrate a way that a translation can be analyzed, despite being two very different paradigms in terms of their executions.
This means that we for example assume that sending and receiving variables is ``free''
($\circ$). In our comparison, work and span costs in \futhark are taken
from the \futhark website\cite{FutharkWorkSpanModel}. The notion of
span is the more interesting of the two, given the potential for
parallelization in \epic.

The first thing to note is that the values in \futhark have a cost and span
of $\mathcal O(1)$, compared to the $\mathcal O(0)$ in the
translation, which could indicate an unacknowledged cost in the
translation.  For arrays and tuples, an improvement in span can be seen
as \epic allows for a full concurrent evaluation of the expressions
inside them.  So instead of span being $S(e_1)+\dots+S(e_n)$ it
becomes $S(\max(e_i))$. The work performed stays the same.

For application, when handling more than one variable the translation
makes use of a tuple input, which then allows for multiple
simultaneous bindings. This can also be done in \futhark and the costs
are the same for both span and work.

\iot involves lower work and span in the translation, as here only the
evaluation of the sub-expression has a cost. However, the difference
in span compared to that of \futhark is only the absence of a single
constant. The cost of \map is the same in both languages, as \futhark
also all handles all the array members in parallel. \textsf{reduce}
can be expressed using \map, \iot, and \texttt{size}, keeping the
asymptotic work and span complexity of ${O}(n)$ and ${O}(log(n))$
respectively that \futhark has.

\begin{table*}
  % \begin{tabular}{|p{0.08\textwidth}|p{0.21\textwidth}|p{0.11\textwidth}|}
  \begin{small}
\begin{tabular}{|l|l|l|}
\hline
\textbf{Construct} & \textbf{Work} & \textbf{Span} \\ \hline\hline
 $\translate{x}_o$   &       $O(0)$     &  $O(0)$      \\ \hline
 $\translate{v}_o$   &       $O(0)$     &    $O(0)$        \\ \hline
$\translate{\texttt{if}(\dots)}_{o}$  &    $O(1+W(\translate{e_1}_o)+\max(W(\translate{e_2}_o),W(\translate{e_3}_o)))$        &     $O(1+S(\translate{e_1}_o)+\max(S(\translate{e_2}_o),S(\translate{e_3}_o)))$       \\ \hline
   $\translate{\lambda x.e}_o$    &    $O(0)$        &    $O(0)$        \\ \hline
  $\translate{e_1\ e_2}_o$     &    $O(1+W_f(\translate{e_1}_o)+W(\translate{e_2}_o))$        &    $O(1+S_f(\translate{e_1}_o)+S(\translate{e_2}_o))$        \\ \hline
  Array  & $O(\sum_{i=1}^n(W(\translate{e_i}_o)))$           &   $O(S(\max(\translate{e_i}_o)))$          \\ \hline
    Tuple  & $O(\sum_{i=1}^n(W(\translate{e_i}_o)))$           &   $O(S(\max(\translate{e_i}_o)))$          \\ \hline
    $\translate{e_1[e_2]}_o$ &    $O(1+W(\translate{e_1}_o)+W(\translate{e_2}_o))$ &  $O(1+\max(S(\translate{e_1}_o),S(\translate{e_2}_o))$     \\ \hline
    $\translate{\texttt{size} \ e_1}_o$ &    $O(W(\translate{e_1}_o))$       &   $O(S(\translate{e_1}_o))$     \\ \hline
    $\translate{\iot\ e_1}_o$ &    $O(W(\translate{e_1}_o))$       &   $O(S(\translate{e_1}_o))$          \\ \hline
    $\translate{\map\ e_1}_o$ &    $O(W_a(\translate{e_1}_o)+W_f(\translate{e_1}_o)*n)$       &  $O(S_a(\translate{e_1}_o)+S_f(\translate{e_1}_o))$ \\ \hline
\end{tabular}
\end{small}
\caption{The different complexities of translated expressions, measured by the number of $\bullet$ reductions.}
\end{table*}

\section{Conclusion}

In this paper we have presented the \butf language, a
$\lambda$-calculus with parallel arrays inspired by the \futhark
programming language, and we show a translation of \butf into \epic, a
variant of \pic that uses polyadic communication and broadcast.

Our translation extends the translation from the $\lambda$-calculus to
the \pic due to Milner et al. with the notion of arrays and
involves defining the usual operations on arrays in a process calculus
setting. Our proof of correctness uses a coinductively defined notion
of operational correspondence.
While we proof that the translation is correct in regards to operational correspondence, we do not show that the translation is fully abstract, or that translated programs diverge.

We present a cost model for our version of the \pic in the form of a
classification of reductions -- they can be either important or
administrative. A cost analysis was performed for the translation to
\epic, and its results were compared with the cost for \futhark's
language constructs.  This comparison shows that the map and reduce
operations in \futhark are similar to the fully parallel ones shown here.

\epic uses broadcasting; while this allows us to have a concise
approach that has no counterpart in the $\lambda$-calculus or general
purpose computer instructions means that it might not represent actual
possible performance in the computers which \futhark targets.  Having
broadcast in \epic makes it rather simple to implement array
indexing. It would be interesting to consider an array structure
without the use of broadcast. Here, one must take into account the
result due to Ene and Muntean \cite{DBLP:conf/fct/EneM99} that
broadcast communication is more expressive than point-to-point
communication.

Our translation is not typed; the next step will be to introduce a type
system in \butf and \epic, and extend the translation to also
translate types. Binary session types \cite{gay2005subtyping} would be
a natural candidate to ensure that the channels in the translation
follow a particular protocol.

Furthermore, it is of interest to validate if the translation can be done in the standard \pic without broadcast and composite names.
This would make it possible to relate the translation with other work in the \pic domain.
Broadcasting and composed names as primitives in \epic might also be unrealistic, when considering \epic as an abstraction for real world hardware.
%as this would allow for a much simpler syntax and semantics.
 % It would also be interesting to explore how different evaluation schemes
% other than call-by-value can be emulated in \epic. The parallel nature
% of \epic might correspond well to a \butf semantic with lazy
% evaluation.

\bibliographystyle{eptcs}
% \bibliography{bib/mybib}

\end{document}